\documentclass[]{interact}

\usepackage{epstopdf}
\usepackage{graphics} 
\usepackage{amsmath,amsfonts}
\usepackage{array}
\usepackage{textcomp}
\usepackage{stfloats}
\usepackage{url}
\usepackage{verbatim}
\usepackage{cite}
\usepackage{siunitx}
\usepackage{mathrsfs}
\usepackage{mathtools}
\mathtoolsset{showonlyrefs=true}
\usepackage{bm}
\usepackage{fancybox}
\usepackage{multirow}
\usepackage{appendix}
\usepackage{algorithmic}
\usepackage{algorithm}
\usepackage[hang,small,bf]{caption}
\usepackage[subrefformat=parens]{subcaption}
\usepackage{booktabs}

\theoremstyle{plain}
\newtheorem{theorem}{Theorem}[section]

\newtheorem{proposition}[theorem]{Proposition}

\theoremstyle{definition}

\newtheorem{problem}[theorem]{Problem}

\theoremstyle{remark}
\newtheorem{remark}{Remark}

\renewcommand{\S}{\mathcal{S}}

\newcommand{\G}{\mathbb{G}}
\newcommand{\R}{\mathbb{R}}
\newcommand{\Q}{\mathcal{Q}}
\newcommand{\Z}{\mathbb{Z}}
\newcommand{\K}{\mathcal{K}}
\newcommand{\Kp}{\mathcal{K}_{\mathsf{p}}}
\newcommand{\Ks}{\mathcal{K}_{\mathsf{s}}}

\newcommand{\E}{\mathcal{E}}

\newcommand{\M}{\mathcal{M}}
\newcommand{\C}{\mathcal{C}}
\newcommand{\Gen}{\mathsf{Gen}}
\newcommand{\Enc}{\mathsf{Enc}}
\newcommand{\Dec}{\mathsf{Dec}}
\newcommand{\pk}{\mathsf{pk}}
\newcommand{\sk}{\mathsf{sk}}
\newcommand{\Ecd}{\mathsf{Ecd}}
\newcommand{\Dcd}{\mathsf{Dcd}}

\newcommand{\argmin}{\mathop{\mathrm{arg~min}}\limits}

\newcommand{\Add}{\mathsf{Add}}
\newcommand{\Mult}{\mathsf{Mult}}
\newcommand{\Res}{\mathsf{Res}}
\newcommand{\evk}{\mathsf{evk}}

\begin{document}

\articletype{Research Article}

\title{Confidential FRIT via Homomorphic Encryption}

\author{
\name{Haruki Hoshino$^\textsf{a}$, Jungjin Park$^\textsf{b}$, Osamu Kaneko$^\textsf{a}$, and Kiminao Kogiso$^\textsf{a}$\thanks{CONTACT Kiminao Kogiso. Email: kogiso@uec.ac.jp}}
\affil{$^\textsf{a}$Department of Mechanical and Intelligent Systems Engineering, Graduate School of Informatics and Engineering, The University of Electro-Communications\\ 1-5-1 Chofugaoka, Chofu, Tokyo 1828585, Japan}
\affil{$^\textsf{b}$Cluster II (Emerging Multi-Interdisciplinary Engineering), School of Informatics and Engineering, The University of Electro-Communications\\ 1-5-1 Chofugaoka, Chofu, Tokyo 1828585, Japan}
}

\maketitle

\begin{abstract}
Edge computing alleviates the computation burden of data-driven control in cyber-physical systems (CPSs) by offloading complex processing to edge servers.
However, the increasing sophistication of cyberattacks underscores the
need for security measures that go beyond conventional IT protections and address the unique vulnerabilities of CPSs.
This study proposes a confidential data-driven gain-tuning framework using homomorphic encryption, such as ElGamal and CKKS encryption schemes, to enhance cybersecurity in gain-tuning processes outsourced to external servers.
The idea for realizing confidential FRIT is to replace the matrix inversion operation with a vector summation form, allowing homomorphic operations to be applied.
Numerical examples under 128-bit security confirm performance comparable to conventional methods while providing guidelines for selecting suitable encryption schemes for secure CPS.
\end{abstract}


\section{Introduction}
A Cyber-Physical System (CPS) plays a central role in realizing Industry 4.0. 
A CPS integrates physical systems with information technology, enabling the development of next-generation power grids \cite{li2012multicast}, multi-robot systems \cite{fink2011robust}, and advanced traffic management systems \cite{guo2018simultaneous}. 
A key driver of CPS development is the advancement of cloud and edge computing~\cite{shi2016edge}. 
Edge computing alleviates the computational burden on IoT (Internet of Things) devices by offloading complex processing to edge servers, enabling both immediacy and efficiency. 
Additionally, the rise of Control as a Service (CaaS)~\cite{esen2015control}, a cloud-based control system model, allows for remote updates and flexible customization of control functions, further accelerating CPS adoption. 
With this trend, data-driven control \cite{hjalmarsson1998iterative,soma2004new,campi2002virtual,coulson2019data,berberich2020data} has gained attention as a promising approach for designing and updating controllers using operational data. 
Data-driven gain tuning methods, such as Fictitious Reference Iterative Tuning (FRIT)~\cite{soma2004new,Kaneko13,Kan15,Kan16,Kan24} and Virtual Reference Feedback Tuning (VRFT)~\cite{campi2002virtual}, exemplify this approach, enhancing the adaptability and efficiency of control system optimization.
These methods are often deployed at the edge, where low-latency processing and reduced communication load are important.
However, there are also compelling advantages to performing such computations in the cloud.
The CPS typically consolidates measurement and log data in cloud-based data lakes for analytics and monitoring, making the input--output trajectories required for data-driven gain tuning readily accessible.
This architecture positions cloud-side execution as a practical and scalable solution for data-driven control in large-scale CPSs.

Despite its benefits, the rapid proliferation of CPSs raises critical security and confidentiality concerns. 
As communication networks expand, control systems face growing threats not only from information leakage and tampering but also from direct attacks on physical infrastructure~\cite{sandberg2015cyberphysical}. 
Traditional security measures, such as encrypting communication channels \cite{amin2009safe,pang2011secure} and implementing redundancy-based protections, have proven insufficient against CPS-specific threats, particularly attacks that manipulate control system behavior. 
Among these, poisoning attacks, which introduce malicious noise into operational data to degrade system performance, pose a severe risk \cite{barreno2006can,shafahi2018poison,fan2022survey,ikezaki2023poisoning}. 
In data-driven control, where controllers are adapted based on input-output data, attackers can exploit this dependency to deliberately cause deviation from the intended performance~\cite{baggio2019data,berberich2020data}. 
The increasing sophistication of cyberattacks underscores the need for security measures that go beyond conventional IT protections and address the unique vulnerabilities of CPSs.

To address these threats, homomorphic encryption has emerged as a promising security solution. 
Partially homomorphic encryption such as Paillier~\cite{paillier1999public} and ElGamal~\cite{elgamal1985public}, along with fully homomorphic encryption such as CKKS \cite{cheon2017homomorphic,Cheon18}, BFV \cite{brakerski2012fully,fan2012somewhat}, and BGV \cite{yagisawa2015fully}, enable secure computation on encrypted data. 
This capability allows control algorithms to be executed in cloud environments while preserving data confidentiality, making homomorphic encryption particularly attractive for CPS applications. 
Integrating homomorphic encryption into control engineering has led to studies on CaaS applications, including linear and polynomial control~\cite{kogiso2015cyber}, model predictive control~\cite{alexandru2020towards}, system identification~\cite{Ada24}, and gain tuning~\cite{schluter2022encrypted,Nag24}.
ln particular, system identification~\cite{Ada24} and gain tuning~\cite{schluter2022encrypted,Nag24} entail considerable computational overhead, especially for operations such as matrix inversion over encrypted data.
The gain tuning methods rely on the CKKS encryption scheme, making it difficult to estimate the required number of encrypted operations in advance.
As a result, expensive bootstrapping may become necessary to manage accumulated noise.
Given the limited computational resources of typical edge devices, offloading such computations to the cloud is often a more practical and scalable solution for secure control.
Furthermore, comparative evaluations of encryption schemes in the context of control system design remain scarce, leaving open questions about the better choice of encryption methods for different control scenarios~\cite{schluter2022encrypted}.

To further clarify the motivation for gain tuning over encrypted data, it is helpful to consider representative scenarios in which preserving data confidentiality is essential.
For example, in industrial IoT environments, controller design may require data from multiple factories or production lines, where confidentiality constraints often limit direct data sharing.
In cloud-based control design, gain tuning is often outsourced while encrypted input/output data is preserved to protect sensitive information.
Secure remote tuning is also critical for systems such as UAVs and robotic fleets, where on-site access is limited and communication channels may be insecure.
In these settings, performing gain tuning directly on encrypted data provides a practical and effective solution.

The objective of this study is to propose an encrypted-data-driven gain tuning framework that accommodates different homomorphic encryption schemes, enhancing cybersecurity in data-driven gain tuning.
Specifically, this study investigates the ElGamal and CKKS encryption schemes to develop a method for updating state feedback gains while keeping the control system’s input and output data encrypted.
The proposed method employs secure computation of cofactor expansion to perform matrix inversion, enabling the gain-tuning process to be delegated to an external server without risking information leakage.
To validate the proposed method under a 128-bit security level, two numerical examples are analyzed, focusing on quantization errors and sensitivity differences between the encryption schemes.
By achieving both security and gain-tuning performance, this research contributes to the practical implementation of secure CPS, offering guidelines for selecting encryption schemes in encrypted control systems.
Additionally, the proposed method performs exact matrix inversion via cofactor expansion, which differs from~\cite{Ada24} that relies on iterative approximation over encrypted data.

The contributions of this study are summarized as follows:
i) This study proposes an encrypted-data-driven gain tuning method, referred to as confidential FRIT, implemented using either the ElGamal or CKKS encryption scheme.
ii) The effectiveness of the proposed method is validated through two numerical examples.
By appropriately tuning quantization sensitivity parameters for both the ElGamal and CKKS encryption schemes, the proposed method achieves performance comparable to conventional state feedback gain update methods while ensuring data confidentiality through secure computation.

The structure of this paper is as follows.
\textbf{Section~\ref{sec:pre}} provides the notations and encryption schemes as preliminaries.
\textbf{Section~\ref{sec:prob}} formulates the problem of confidential FRIT.
\textbf{Section~\ref{sec:cFRIT}} proposes the confidential FRIT methods based on the ElGamal or CKKS encryption scheme.
\textbf{Section~\ref{sec:exp}} demonstrates the effectiveness of the proposed confidential FRIT through two numerical examples.
\textbf{Section~\ref{sec:con}} concludes the paper.

\section{Preliminaries}\label{sec:pre}
\subsection{Notations}
The sets of real numbers, integers, rational numbers, security parameters, key pairs, public keys, private keys, plaintexts, and ciphertexts are represented as $\R$, $\Z$, $\Q$, $\S$, $\K$, $\Kp$, $\Ks$, $\M$, and $\C$, respectively. 
The notations of sets are as follows: 
$\R^-:=\{x\in\R\,|\,x<0\}$, 
$\Z^+\coloneqq\{z\in\Z\,\mid\,0\leq z\}$,
$\Z_n \coloneqq \{z \in \Z \mid 0 \leq z < n\}$, and 
$\Z_{>q}:=\{z\in\Z\,|\,z>q\}$.
The set of vectors of size $n$ is denoted by $\R^n$, and the set of matrices of size $m \times n$ is denoted by $\R^{m \times n}$.  
The $i$-th component of a vector $v$ is denoted by $v_i$, and its $\ell_2$-norm is represented as $\lVert v \rVert$.  
The $(i,j)$-th component of a matrix $X$ is denoted by~$X_{ij}$. 

For a regular matrix $X\in\R^{n\times n}$, its inverse $X^{-1}$ is expressed as  
\begin{align*}
X^{-1} 
= \frac{1}{\vert X \vert} \begin{bmatrix}
x_{11} & x_{21} & \cdots & x_{n1} \\
x_{12} & x_{22} & \cdots & x_{n2} \\
\vdots & \vdots & \ddots & \vdots \\
x_{1n} & x_{2n} & \cdots & x_{nn}
\end{bmatrix},
\end{align*}
where $x_{ij}$ is the $(i,j)$ cofactor of the matrix $X$.  
For a permutation $\sigma = \begin{pmatrix}
  1 & 2 & 3 & \cdots & n \\
  i_1 & i_2 & i_3 & \cdots & i_n
\end{pmatrix}$, $\mathrm{sgn}\,\sigma=1$ if $\sigma$ is an even permutation; otherwise, -1 returns.
The determinant of a matrix $X\in\R^{n\times n}$ is expressed as $
  \vert X\vert=\sum_{\sigma\in\mathbb{S}^n} 
  \big\{(\mathrm{sgn}\,\sigma)\prod_{i=1}^n x_{i\sigma(i)}\big\}$,
where $\sigma$ represents a permutation and $\mathbb{S}^n$ is the symmetric group of degree $n$ (the automorphism group of $\{1,\cdots,n\}$).

The polynomial residue ring is denoted as $\mathcal{R}=\Z[X]/(X^d+1)$.
The set of all polynomials in $\mathcal{R}$ whose coefficients belong to $\{z\in\Z\mid -\frac{n}{2}<z\leq \frac{n}{2}\}$ is denoted as~$\mathcal{R}_n$. 
The symbol $\bmod$ represents the residue operator. 
For a polynomial $a$ and an integer $N$, $a\bmod N$ takes the residue of each coefficient of $a$ modulo $N$. 
For two polynomials $a$ and $b$, the notation $a\cdot b$ represents the polynomial multiplication in the polynomial residue ring $\mathcal{R}$. 

\subsection{ElGamal Encryption}\label{ElGamal}
The ElGamal encryption scheme~\cite{elgamal1985public}, which is public-key multiplicatively homomorphic encryption, is defined as a tuple $\E_E\coloneqq (\Gen,\Enc,\Dec)$. 
$\Gen:\S\to\K=\Kp\times\Ks$
$:k\mapsto (\pk,\sk)=((p,q,g,h),s)$,
$\Enc:\Kp\times\M\to\C$,
$:(\pk,m)\mapsto c=(g^r\bmod p,mh^r\bmod p)$, and
$\Dec:\Ks\times\C\to\M$,
$:(\sk,(c_1,c_2))\mapsto c_1^{-s}c_2\bmod p$, 
which are the key generation, encryption, and decryption algorithms, respectively.
$\pk\in\Kp$ and $\sk\in\Ks$ are the public and private keys.
$q$ is a $k$-bit prime, and $p=2q+1$ is a safe prime.
$g$ is the generator of a cyclic group $\G\coloneqq\{g^i\bmod p\,\vert\,i\in\Z_q\}=\M\subset\Z_p\backslash\{0\}$, 
where 
$g^q\bmod p=1$, 
$h=g^s\bmod p$, 
$\M=\G$,
$\C=\G^2$, and 
$r$ and $s$ are random values chosen from $\Z_q$.
Note that for $m$, $m'\in\M$, the ElGamal encryption scheme possesses multiplicative homomorphism: 
$\Dec(\sk,\Enc(\pk,m)\ast\Enc(\pk,m')\bmod p)=mm'\bmod p$, where $\ast$ denotes the Hadamard product.
Hereafter, for simplicity, the arguments $\pk$ and $\sk$ of $\Enc$ and $\Dec$ will be omitted.

Furthermore, to integrate the encryption scheme into control systems, the encoder $\Ecd_{\gamma_e}:\R\to\M$ and decoder $\Dcd_{\gamma_e}:\M\to\R$ are needed~\cite{teranishi2019stability}:
\begin{align}
\Ecd_{\gamma_e}:x\mapsto\check{x} 
&=\mathrm{min}\{\argmin_{m\in\M} \vert x\gamma_e^{-1}+p\textbf{1}_{\R_{<0}}(x)-m\vert\}
=x\gamma_e^{-1}+p\textbf{1}_{\R_{<0}}(x)+\delta,\label{ecd_e}\\
\Dcd_{\gamma_e}:\check{x}\mapsto\bar{x}&=\gamma_e(\check{x}-p\textbf{1}_{\R_{>q}}(\check{x})), \label{dcd_e}
\end{align}
where $\gamma_e\in (0,1)$ is the sensitivity parameter ($\gamma_e^{-1}$ is also called a quantization gain); given set $S$, $\textbf{1}_{S}(x)=1$ if $x\in S$; otherwise, it returns 0. 
In this case, the rounding error over the plaintext space is denoted as $\delta:=\delta(x,\gamma_e)=|\check{x}-x\gamma_e^{-1}-p\textbf{1}_{\R_{<0}}(x)|\in\Z^+$.
The encryption-induced quantization error is denoted as $x-\bar{x}$ through the operation $\Dcd_{\gamma_e}\circ\Ecd_{\gamma_e}$.

\subsection{CKKS Encryption}
The CKKS encryption scheme~\cite{cheon2017homomorphic}, which is a fully homomorphic encryption scheme, is defined as a tuple $\E_C\coloneqq(\Gen,\Enc,\Dec,\Res,\Add,\Mult)$. 
$\Gen:(d,L,\gamma_c,q_0)\mapsto(\sk,\pk,\evk)$,
$\Enc:(\check{x},\pk)\mapsto x^{ct}$,
$\Dec:(x^{ct},\sk)\mapsto\check{x}$,
$\Res:x^{ct}_l=(ct_{10},ct_{11})\in\mathcal{R}^2_{q_{l}}\mapsto x^{ct}_{l-1}=(ct_{20},ct_{21})\in\mathcal{R}^2_{q_{l-1}}$,
$\Add:(x^{ct}_1,x^{ct}_2)\mapsto x_{\textrm{add}}^{ct}=x^{ct}_1\oplus x^{ct}_2$, and 
$\Mult:(x^{ct}_1,x^{ct}_2)\mapsto x^{ct}_{\textrm{mul}}=x^{ct}_1\otimes x^{ct}_2$ are key generation, encryption, decryption, rescaling algorithms, and addition and multiplication operations over encrypted data, respectively.
$d\in\Z$ is the degree of the polynomial, 
$L\in\Z$ is the level, 
$\gamma_c\in(0,1)$ is the sensitivity parameter ($\gamma_c^{-1}$ is also called the scale parameter), 
$q_0\in\R$ is chosen to be sufficiently large compared to $\gamma_c^{-1}$, and  
$\pk$, $\sk$, and $\evk$ are the public, private, and evaluation keys, respectively.
The CKKS encryption~\cite{cheon2017homomorphic,Cheon18} can handle complex numbers to be encrypted, and its scheme includes processes of encoding and decoding:
\begin{align}
\Ecd_{\gamma_c}:x\mapsto\check{x}=\lfloor x\gamma_c^{-1}\rceil,\quad 
\Dcd_{\gamma_c}:\check{x}\mapsto\bar{x}=\lfloor\check{x}\gamma_c\rceil,
\label{ecd}
\end{align}
To simplify the discussion, this study uses only real numbers and separates the encoding and decoding processes from the CKKS algorithms.
In this case, the plaintext, ciphertext, and key spaces are defined as $\M=\Z_{\pm n}:=\{z\in\Z\mid-n/2<z\leq n/2\}$, $\C:=\mathcal{R}_{q_l}^2$, and $\K:=\mathcal{R}$, respectively.
Similarly, the precision of the plaintext during computation or decryption can be specified by the magnitude of $\gamma_c$.
The details of each of the algorithms will be introduced in \textbf{Appendix~\ref{app:ckks}}.

\section{Confidential FRIT Problem}\label{sec:prob}
This section presents the confidential FRIT setup, which involves developing a gain-tuning method performed over encrypted data of state-feedback control systems using homomorphic encryption, allowing to outsource secure computation processing to an external server.  

A linear plant is described by the discrete-time state-space representation:
\begin{align}
x(k+1)=Ax(k)+Bu(k), \label{plant}
\end{align}
where 
$k\in\Z^+$, $x\in\R^n$, and $u\in\R$ are the step, state (measurable), control input, and output, respectively.
The initial state is assumed to be zero, i.e., $x(0)=0$.
The system order for minimal realization is assumed to be known, and the plant model \eqref{plant} is assumed to be unknown to a control system designer.
For the plant \eqref{plant}, we consider the following controller’s structure, designed to achieve $\lim_{k\to\infty}x(k)=0$,
\begin{align}
u(k)=Fx(k)+v(k),\label{fcl}
\end{align}
where $F\in\R^{1\times n}$ is the state feedback gain, and 
$v\in\R$ is an intentionally applied signal used by the designer to evaluate the closed-loop system: $H(F)\coloneqq\frac{x(z)}{v(z)}=(zI-A-BF)^{-1}B\in\R^n[z]$.
The control system designer (client) is assumed to have access to the control system, its data of $u$ and $x$, and an initial gain $F=F_{\textrm{ini}}$ and aims to obtain a new gain $F^*$, achieving the desired closed-loop characteristics $H_d$. 
To do so, the client collects the encrypted dataset $\mathcal{D}$ using a specific homomorphic encryption scheme to outsource securely and receives the encrypted dataset $\mathcal{F}$, which corresponds to the extracted new gain $F^*_\mathcal{E}$, as shown in Fig.~\ref{fig:setup}.
In this setup, no private keys of the homomorphic encryption scheme are sent to the server, ensuring that the decryption process is not performed as part of the computation of $\mathcal{F}$ on the server.
Thus, this study considers the following problem referred to as the confidential FRIT problem.

\begin{problem}
For the feedback control system consisting of the plant \eqref{plant} and the controller \eqref{fcl}, clarify the encrypted datasets $\mathcal{D}$ and $\mathcal{F}$, as well as a confidential FRIT algorithm based on homomorphic operations, such that the feedback gain $F^*_\mathcal{E}$ extracted from $\mathcal{F}$ is equivalent to that obtained by the conventional FRIT algorithm~\cite{soma2004new}.
Additionally, the confidential FRIT algorithm is implemented using either ElGamal or CKKS encryption.
\end{problem}

The confidential FRIT problem is a secure implementation of the conventional FRIT~\cite{soma2004new}, using homomorphic encryption to enhance cybersecurity by mitigating the risks of the transmitted datasets being modified by an adversary or processed data being tampered with by a semi-honest server administrator.
This approach enables the secure updating of the state feedback gain using the retained data while outsourcing necessary computations to an external server, effectively preventing threats such as eavesdropping or unauthorized access.

\begin{figure}[t!]
\centering
\includegraphics[width=0.6\linewidth]{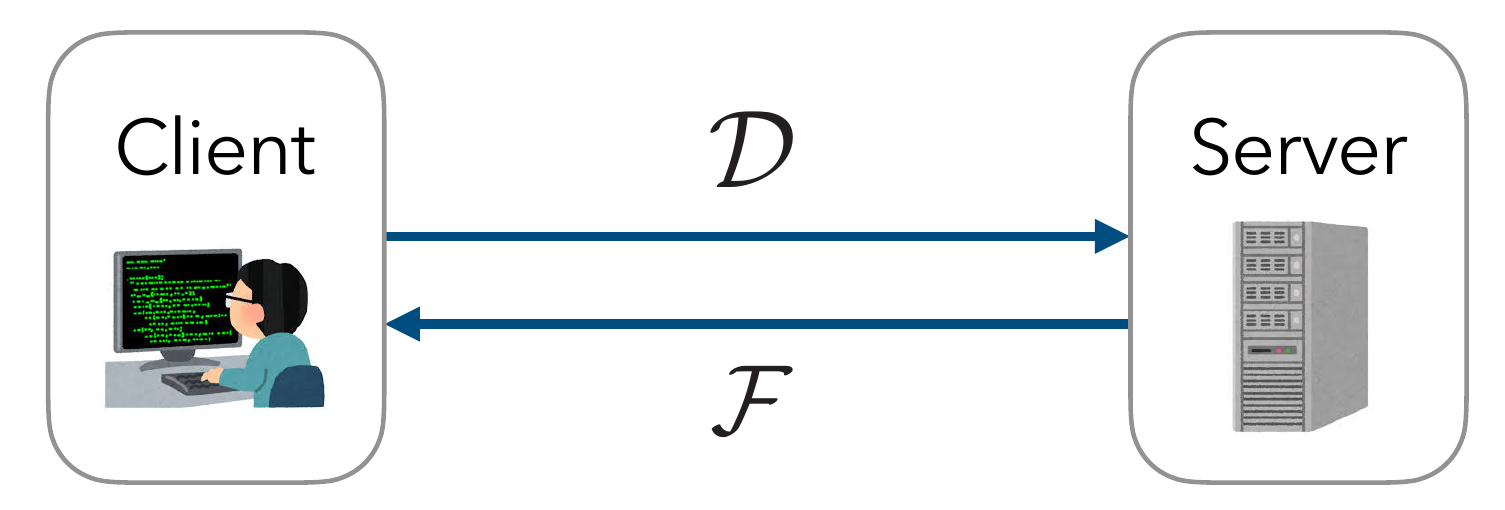}
\caption{Setups of conventional and confidential FRITs, which handle unencrypted and encrypted datasets ($\mathcal{D},\mathcal{F}$), respectively.}
\label{fig:setup}
\end{figure}

\section{Proposed Confidential FRIT}\label{sec:cFRIT}
This section proposes a confidential framework for updating state-feedback gains over encrypted data using either the ElGamal or CKKS encryption scheme.
The idea of the proposed method is to replace matrix operations involving the conventional FRIT algorithm for matrices $\Gamma\in\R^{nN}$ and $W\in\R^{nN\times n}$:
\begin{align}
F^*=-\Gamma^\top W(W^\top W)^{-1}, \label{original}
\end{align}
with the homomorphic operations defined in the encryption schemes, where $\Gamma$ and $W$ are defined using $N$ signal data points in terms of $x$ and $u$ in \eqref{plant} and \eqref{fcl}, which will be explained in \textbf{Appendix~\ref{app:frit}}.
Specifically, the multiplicative homomorphism of the ElGamal encryption scheme or the additive and multiplicative homomorphisms of the CKKS encryption scheme is used to perform calculations while keeping the data encrypted.
This idea motivates us to derive an alternative expression of \eqref{original} that avoids the direct computation of the inverse matrix.
Focusing on the fact that each element of the inverse matrix can be represented in terms of cofactors, we can obtain the following main theorem.

\begin{theorem}\label{thm1}
The equation \eqref{original} is expressed as follows:
\begin{align}
F^*=
-\sum_{k=1}^{(n-1)!}\sum_{i=1}^{nN}\sum_{l=1}^n\Gamma_iW_{il}\Phi_{k,l}=\sum_{j=1}^MF_j^*,
\label{F2}
\end{align}
where $\Gamma_i\in\R$ is the $i$-th element of the vector $\Gamma\in\R^{nN}$,  
$W_{il}\in\R$ is the $(i,l)$-th element of the matrix $W\in\R^{nN\times n}$, 
$\Phi_{k,l}\in\R^{1\times n}$ is the $l$-th row vector of the intermediate generation matrix $\Phi_k\in\R^{n\times n}$ such that $(W^\top W)^{-1}=\sum_{k=1}^{(n-1)!}\Phi_k$, $F^*_j\in\R^{1\times n}$, $\forall j$, and 
$M=(n-1)!Nn+(n-1)Nn$.
\end{theorem}
\begin{proof}
Defining $\Phi_{k,ji}:=\lvert\Psi\rvert^{-1}(-1)^{i+j}\tilde{\Psi}_{k,ij}$ with $\Psi=W^\top W$, we show that for the $(j,i)$ element $\Phi_{k,ji}$ of the intermediate generation matrix $\Phi_k$,  
$\sum_{k=1}^{(n-1)!}\Phi_{k,ji}=\lvert\Psi\rvert^{-1}\tilde{\Psi}_{ij}$ holds, 
where $\tilde{\Psi}_{ij}$ is the matrix obtained by removing the $i$-th row and $j$-th column of the matrix $\Psi$.
Expanding the determinant calculation,
\begin{align}
\sum_{k=1}^{(n-1)!}\Phi_{k,ji}&=\sum_{k=1}^{(n-1)!}\lvert\Psi\rvert^{-1}(-1)^{i+j}\tilde{\Psi}_{k,ij}
=\lvert\Psi\rvert^{-1}(-1)^{i+j}\sum_{k=1}^{(n-1)!}\mathrm{sgn}\sigma_k\prod_{l=1}^{n-1}\tilde{\Psi}_{l,\sigma_k(l)}\\
&=\lvert\Psi\rvert^{-1}(-1)^{i+j}\lvert\tilde{\Psi}_{ij}\rvert =\lvert\Psi\rvert^{-1}\tilde{\Psi}_{ij},
\end{align}
where $\mathrm{sgn}\sigma_k$ represents the sign of the permutation $\sigma_k$, which corresponds to the $k$-th row of the signed permutation matrix $\Sigma$:
\begin{align}
\Sigma=
\begin{bmatrix}
\sigma_1(1) & \sigma_1(2) & \cdots & \sigma_1(n-1) & \mathrm{sgn}\sigma_1 \\
\sigma_2(1) & \sigma_2(2) & \cdots & \sigma_2(n-1) & \mathrm{sgn}\sigma_2 \\
\vdots & \vdots & \ddots & \vdots & \vdots \\
\sigma_{(n-1)!}(1) & \sigma_{(n-1)!}(2) & \cdots & \sigma_{(n-1)!}(n-1) & \mathrm{sgn}\sigma_{(n-1)!}
\end{bmatrix}\in\mathbb{Z}^{(n-1)!\times n}.
\label{sigma}
\end{align}
Accordingly, the following holds:
\begin{align}
\sum_{k=1}^{(n-1)!} \Phi_k &=\begin{bmatrix}
\lvert\Psi\rvert^{-1}\tilde{\Psi}_{11} & \lvert\Psi\rvert^{-1}\tilde{\Psi}_{21} & \cdots & \lvert\Psi\rvert^{-1}\tilde{\Psi}_{n1}\\
\lvert\Psi\rvert^{-1}\tilde{\Psi}_{12} & \lvert\Psi\rvert^{-1}\tilde{\Psi}_{22} & \cdots & \lvert\Psi\rvert^{-1}\tilde{\Psi}_{n2}\\
\vdots & \vdots & \ddots & \vdots\\
\lvert\Psi\rvert^{-1}\tilde{\Psi}_{1n} & \lvert\Psi\rvert^{-1}\tilde{\Psi}_{2n} & \cdots & \lvert\Psi\rvert^{-1}\tilde{\Psi}_{nn}
\end{bmatrix}
=\frac{1}{\lvert\Psi\rvert}\tilde{\Psi}=\Psi^{-1}.
\end{align}
Thus, the inverse matrix operation can be replaced by the summation of intermediate generation matrices $\Phi_k$, leading to
$(W^\top W)^{-1}=\sum_{k=1}^{(n-1)!}\Phi_k$.
Using this replacement, we obtain: 
\begin{align}
F^*&=-\Gamma^\top W(W^\top W)^{-1}
=-\Gamma^\top W\sum_{k=1}^{(n-1)!}\Phi_k
=-\sum_{k=1}^{(n-1)!}\Gamma^\top W\Phi_k,\\
&=-\sum_{k=1}^{(n-1)!}\left(\sum_{i=1}^{nN}\Gamma_i W_i\Phi_k\right)
=-\sum_{k=1}^{(n-1)!}\sum_{i=1}^{nN}\Gamma_i W_i\Phi_k,\\
&=-\sum_{k=1}^{(n-1)!}\sum_{i=1}^{nN}\left(\Gamma_i \sum_{l=1}^nW_{il}\Phi_{k,l}\right)
=-\sum_{k=1}^{(n-1)!}\sum_{i=1}^{nN}\sum_{l=1}^n\Gamma_iW_{il} \Phi_{k,l},\label{F_j}
\end{align}
where $W_i\in\R^n$ is the $i$-th row vector of $W$.
Defining $F^*_j:=-\Gamma_i W_{il}\Phi_{k,l}$ for $j=(k-1)nN+(i-1)n+l$, we obtain $F^*=\sum_{j=1}^MF^*_j$, where $M>0$.
The parameter $j$ increases sequentially from $1$ to $M$, and is defined as:
$j(k,i,l)=(k-1)nN+(i-1)n+l$, 
$\forall k\in\{1,2,\dots,(n-1)!\}$,
$i\in\{1,2,\dots,nN\}$, and 
$l\in\{1,2,\dots,n\}$.
Thus, the minimum and maximum values of $j$ are $j(1,1,1)=1$ and $j((n-1)!,nN,n)=((n-1)!-1)nN+(nN-1)n+n=M$, repspectively.
This terminates the proof.
\end{proof}

\begin{proposition}\label{prop}
Using the multiplicative homomorphism of ElGamal encryption and the additive and multiplicative homomorphisms of CKKS encryption, the secure computation of \eqref{F2} is realized by each of the following equations, respectively:
\begin{align}
\bar{F}_\mathrm{E}^*&:=\Dcd_{\gamma_e}\left(\sum_{k=1}^{(n-1)!}\sum_{i=1}^{nN}\sum_{l=1}^{n}
\Dec\left(\Enc(\check{-1})\ast\Enc(\check{\Gamma}_i)*\Enc(\check{W}_{il})*\Enc(\check{\Phi}_{k,l})\right)\right),\label{FE1}\\
&=\Dcd_{\gamma_e}\left(\sum_{j=1}^{M}\Dec\left(\Enc(\check{F}_j^*)\right)\right),\label{FE2}\\
\bar{F}_\mathrm{C}^*&:=\Dcd_{\gamma_c}\left(\Dec\left(\bigoplus_{k=1}^{(n-1)!}\bigoplus_{i=1}^{nN}\bigoplus_{l=1}^{n}
\Enc(\check{-1})\otimes\Enc(\check{\Gamma}_i)\otimes\Enc(\check{W}_{il})\otimes\Enc(\check{\Phi}_{k,l})\right)\right),\label{FC1}\\
&=\Dcd_{\gamma_c}\left(\Dec\left(\bigoplus_{j=1}^{M}\Enc(\check{F}_j^*)\right)\right),\label{FC2}
\end{align}
where the operator $\bigoplus$ denotes homomorphic addition over encrypted data, defined as,
\begin{align}
\bigoplus_{i=1}^n\Enc(m_i):=\Enc(m_1)\oplus\Enc(m_2)\oplus\cdots\oplus\Enc(m_n),\ \  \forall m_i\in\M. 
\end{align}
Furthermore, we assume that the quantization error is sufficiently small to be negligible.
Then, there exists a sufficiently small value $\epsilon\geq 0$ such that $||\bar{F}_{\mathrm{E}}^*-F^*||<\epsilon$ and $||\bar{F}_{\mathrm{C}}^*-F^*||<\epsilon$ hold, respectively.
\end{proposition}

\begin{proof}
Since the operations in \eqref{F2} in \textbf{Theorem~\ref{thm1}} are simply replaced with their corresponding homomorphic counterparts, the proof is omitted.
\end{proof}

\begin{remark}
The computational burden increases super-exponentially with the state dimension $n$ of the plant, since the upper limit in \eqref{FE2} and \eqref{FC2} involves $(n-1)!$, respectively.
While the proposed algorithms are suitable for offline use in plants with small to moderate dimensions, the computation cost becomes prohibitive for high-dimensional systems due to the factorial growth in complexity. 
To address this issue, future work will consider approximation-based methods for computing the inverse matrix to reduce the computational load.
\end{remark}

\subsection{ElGamal-based Confidential FRIT}
We present the ElGamal-based confidential FRIT procedure as follows. 
First, the client generates the public and private keys of the ElGamal encryption scheme, as described in \textbf{Section~\ref{ElGamal}}, and prepares the encrypted dataset $\mathcal{D}$ using an appropriate sensitivity parameter $\gamma_e$, 
\begin{align}
\mathcal{D}:=\left\{
\Enc(\check{\Gamma}),
\Enc(\check{W}),
\Enc(\check{\Psi}),
\Enc(\rvert\check{\Psi}\lvert^{-1}),
\pk\right\},
\end{align} 
where $\check{\Psi}=\Ecd_{\gamma_e}(W^\top W)$. 
The server then generates an encrypted dataset $\mathcal{F}$ from $\mathcal{D}$ using \textbf{Algorithm~\ref{elgamal_server_processing}}, which handles the argument of the $\Dec$ function in~\eqref{FE1},
\begin{align}
\mathcal{F}:=\left\{
\Enc(\check{F}_j^*),
\xi_j, \forall j\in\{1,2,\cdots,M\}\right\}.
\end{align}
The client uses the received dataset $\mathcal{F}$ to compute $\bar{F}^*_\mathrm{E}$ according to \textbf{Algorithm~\ref{elgamal_client_processing}}, which implements the operation~\eqref{FE2}.
Here, $\xi_j\in\Z^n$ is the element of $\mathcal{F}$ that stores the exponent information related to $\gamma_e$, required for decoding $\check{F}_j^*$.
In addition, \textbf{Algorithm~\ref{elgamal_phi_omega_generation}}, which is called within \textbf{Algorithm~\ref{elgamal_server_processing}}, computes the determinant required for the proof of \textbf{Theorem~\ref{thm1}}, reflecting the ElGamal-specific definition of $\mathcal{D}$ and $\mathcal{F}$.

Next, we analyze the relationship between the sensitivity parameter and the quantization error for ElGamal-based confidential FRIT, where the product of $n$ multiplications is considered.

\begin{proposition}\label{theo}
Consider the ElGamal encryption scheme with a $k$-bit prime $q$ and $p=2q+1$. 
Let $x_i\in\R$, $\forall i\in\{1,2,\cdots,n\}$, denote scalar values to be multiplied.
For a given sensitivity parameter $\gamma_e$, suppose rounding errors $\delta(x_i,\gamma_e)=|\check{x}_i-x_i\gamma_e^{-1}-p\textbf{1}_{\R_{<0}}(x_i)|\in\Z^+$ satisfy the inequalities, 
$-q-1<\gamma_e^{-n}\prod_{i=1}^n(x_i+\gamma_e\delta_i)\leq q$. Then, the following equation holds:
\begin{align}
\Dcd_{\gamma_e^n}\left(\Dec\left(\overset{n}{\prod}_{i=1}^*\Enc\left(\Ecd_{\gamma_e}(x_i)\right)\right)\right)=\prod_{i=1}^n\left(x_i+\gamma_e\delta_i\right),\label{thm:eq}
\end{align}
where $\check{x}_i=\Ecd_{\gamma_e}(x_i)$, $\forall i\in\{1,2,\cdots,n\}$, and $\delta_i:=\delta(x_i,\gamma_e)$.
Furthermore, as $\gamma_e\rightarrow 0$, the left-hand side of \eqref{thm:eq} converges to $\Pi_{i=1}^nx_i$.
\end{proposition}
\begin{proof}
Due to the correctness and multiplicative homomorphism of the ElGamal encryption scheme, the left-hand side of \eqref{thm:eq} can be rewritten as: $\Dcd_{\gamma_e^n}(\overset{n}{\prod}_{i=1}\Ecd_{\gamma_e}(x_i))$.
For $x_i>0$ and $\check{x}_i<q$, the encoding process \eqref{ecd_e}
and the decoding process \eqref{dcd_e} yield: $\Ecd_{\gamma_e}(x_i)=\gamma_e^{-1}x_i+\delta_i=\gamma_e^{-1}(x_i+\gamma_e\delta_i)$ and 
\begin{align}
\Dcd_{\gamma_e^n}\left(\prod^n_{i=1}\Ecd_{\gamma_e}(x_i)\right)
=\Dcd_{\gamma_e^n}\left(\prod^n_{i=1}\gamma_e^{-1}(x_i+\gamma_e\delta_i)\right)
=\gamma_e^n\left(\prod^n_{i=1}\gamma_e^{-1}(x_i+\gamma_e\delta_i)\right),
\end{align}
which corresponds to the right-hand side of \eqref{thm:eq}.
Similarly, for $x_i<0$ and $\check{x}_i>q$, the same reasoning holds.
This terminates the proof.
\end{proof}

\textbf{Propositions~\ref{prop}} and \textbf{\ref{theo}} imply that, as $\gamma_e\rightarrow 0$, the impact of encryption-induced quantization errors on the product can be made negligible to achieve $\bar{F}_{\mathrm{E}}^*=F^*$, even though the values of $\delta_i$, $\forall i$, remain unknown.
 This property will be examined numerically in \textbf{Section~\ref{sec:exp}}.
 
\begin{figure}[t!]
  \begin{algorithm}[H]
  \caption{ElGamal-Based Server-Side Confidential FRIT}
  \label{elgamal_server_processing}
  \begin{algorithmic}[1]
  \REQUIRE $\Enc(\check{\Gamma})$, $\Enc(\check{W})$, $\Enc(\check{\Psi})$, $\Enc(\lvert\check{\Psi}\rvert^{-1})$, and $\pk$.
  \ENSURE $\Enc(\check{F}^*_j)$, $\xi_j$, $\forall j\in\{1,2,\cdots,M\}$.
  \STATE Compute $\Sigma$ of~\eqref{sigma}.
  \STATE Compute $\Enc(\check{\Phi}_k)$, $\Omega_k\in\R^{n\times n}$, $k\in\{1,2,\cdots,(n-1)!\}$, using \textbf{Algorithm~\ref{elgamal_phi_omega_generation}}.
  \FOR{$k=1$ to $(n-1)!$}
  \FOR{$i=1$ to $nN$}
  \FOR{$l=1$ to $n$}
  \STATE $j\gets(k-1)nN+(i-1)n+l$
  \STATE $\Enc(\check{F}^*_j)\gets\Enc(\check{-1})*\Enc(\check{\Gamma}_i)*\Enc(\check{W}_{il})*\Enc(\check{\Phi}_{k,l})$
    \STATE $\xi_{j,l}\gets\Omega_{k,l}+\begin{bmatrix}3&3&\cdots&3 \end{bmatrix}$
  \ENDFOR
  \ENDFOR
  \ENDFOR
  \end{algorithmic}
  \end{algorithm}
\end{figure}

\begin{figure}[t!]
  \begin{algorithm}[H]
  \caption{ElGamal-Based Client-Side Confidential FRIT}
  \label{elgamal_client_processing}
  \begin{algorithmic}[1]
  \REQUIRE $\Enc(\check{F}^\ast_j)$, $\xi_j$, $\forall j\in\{1,2,\cdots,M\}$
  \ENSURE $\bar{F}^*_\mathrm{E}$
  \FOR{$j=1$ to $M$}
  \STATE $\check{F}^*_j\gets\Dec(\Enc(\check{F}^*_j))$.
  \FOR{$l=1$ to $n$}
  \STATE $\gamma\gets(\gamma_e)^{\xi_{j,l}}$
  \STATE $\bar{F}^*_{j,l}\gets\Dcd_{\gamma}(\check{F}^*_{j,l})$.
  \ENDFOR
  \ENDFOR
  \STATE $\bar{F}^*_\mathrm{E}\gets\sum_{j=1}^{M}\begin{bmatrix}\bar{F}^*_{j,1}& \bar{F}^*_{j,2} & \cdots & \bar{F}^*_{j,n}\end{bmatrix}$
  \end{algorithmic}
  \end{algorithm}
\end{figure}

\begin{figure}[t!]
  \begin{algorithm}[H]
  \caption{Computation of $\Enc(\check{\Phi}_k)$ and $\Omega_k$ for ElGamal-based Confidential FRIT}
  \label{elgamal_phi_omega_generation}
  \begin{algorithmic}[1]
  \REQUIRE $\Enc(\check{\Psi})$, $\Enc(\lvert\check{\Psi}\rvert^{-1})$, $\Sigma$, and $\pk$
  \ENSURE $\Enc(\check{\Phi}_k)$, $\Omega_k\in\R^{n\times n}$, $\forall k\in\{1,2,\cdots,(n-1)!\}$
  \FOR{$k=1$ to $(n-1)!$}
  \FOR{$i=1$ to $n$}
  \FOR{$j=1$ to $n$}
  \STATE Compute $\Enc(\check{\tilde{\Psi}}_{ij})$ from $\Enc(\check{\Psi})$.
  \IF{$\mathrm{sgn}\sigma_k=-1$}
  \STATE
  $\Enc(\check{\tilde{\Psi}}_{k,ij})\gets\Enc(\check{-1})*{\prod^*}_{l=1}^{n-1}\Enc(\check{\tilde{\Psi}}_{l\sigma_k(l)})$.
  \STATE
  \(
    \Omega_{k,ji}\gets 
    \begin{cases}
      (n-1)+1+2, & \mathrm{if}\ \ (-1)^{i+j}=-1,\\
      (n-1)+1+1, & \mathrm{otherwise.}
    \end{cases}
  \)
  \ELSE
  \STATE
  \(
    \Enc(\check{\tilde{\Psi}}_{k,ij})\gets{\prod^*}_{l=1}^{n-1}\Enc(\check{\tilde{\Psi}}_{l\sigma_k(l)}).
  \)
  \STATE
  \(
    \Omega_{k,ji} \gets 
    \begin{cases}
      (n-1)+2, & \mathrm{if}\ \ (-1)^{i+j}=-1,\\
      (n-1)+1, & \mathrm{otherwise.}
    \end{cases}
  \)
  \ENDIF
  \IF{$(-1)^{i+j}=-1$}
  \STATE
  \(
\Enc(\check{\Phi}_{k,ji})\gets\Enc(\lvert\check{\Psi}\rvert^{-1})*\Enc(\check{-1})*\Enc(\check{\tilde{\Psi}}_{k,ij}).
  \)
  \ELSE
  \STATE
  \(
\Enc(\check{\Phi}_{k,ji})\gets\Enc(\lvert\check{\Psi}\rvert^{-1})*\Enc(\check{\tilde{\Psi}}_{k,ij}).
  \)
  \ENDIF
  \ENDFOR
  \ENDFOR
  \ENDFOR
  \end{algorithmic}
  \end{algorithm}
\end{figure}

\subsection{CKKS-based Confidential FRIT}
We present the CKKS-based confidential FRIT procedure as follows. 
The client prepares public, private, and evaluation keys for the CKKS encryption scheme, along with the encrypted dataset~$\mathcal{D}$, based on an appropriate sensitivity parameter $\gamma_c$:
\begin{align}
\mathcal{D}:=\left\{
\Enc(\check{\Gamma}),
\Enc(\check{W}),
\Enc(\rvert\check{\Psi}\lvert^{-1}),
\pk,\evk\right\},
\end{align}
where $\check{\Psi}=\Ecd_{\gamma_c}(W^\top W)$. 
The server generates the following encrypted dataset $\mathcal{F}$ from $\mathcal{D}$ using \textbf{Algorithm~\ref{ckks_server_processing}}, which handles the argument of the $\Dec$ function in~\eqref{FC1}:
\begin{align}
\mathcal{F}:=\{\Enc(\check{F}^*)\}.
\end{align}
The client then uses the received dataset $\mathcal{F}$ to compute $\bar{F}^*_\mathrm{C}$ based on \textbf{Algorithm~\ref{ckks_client_processing}}, which implements the operation~\eqref{FC2}.
Similarly, \textbf{Algorithm~\ref{ckks_phi_omega_generation}}, invoked within \textbf{Algorithm~\ref{ckks_server_processing}}, performs the determinant calculation required in the proof of \textbf{Theorem~\ref{thm1}}, using the CKKS-based definition of $\mathcal{D}$ and $\mathcal{F}$.
Additionally, analysis of the CKKS-based FRIT is omitted due to the complexity of the CKKS encryption scheme~\cite{Cheon18}, implemented into the Microsoft SEAL library.
Instead, the quantization error will be evaluated numerically in \textbf{Section~\ref{sec:exp}}.

\begin{figure}[t!]
  \begin{algorithm}[H]
  \caption{CKKS-Based Server-Side Confidential FRIT}
  \label{ckks_server_processing}
  \begin{algorithmic}[1]
  \REQUIRE $\Enc(\check{\Gamma})$, $\Enc(\check{W})$, $\Enc(\lvert\check{\Psi}\rvert^{-1})$, $\pk$, and $\evk$
  \ENSURE $\Enc(\check{F}^*)$
  \STATE Compute $\Sigma$ of~\eqref{sigma}.
  \STATE Compute $\Enc(\check{\Phi}_k)$, $\forall k\in\{1,2,\cdots,(n-1)!\}$, using \textbf{Algorithm~\ref{ckks_phi_omega_generation}}.
  \FOR{$k = 1$ to $(n-1)!$}
  \FOR{$i = 1$ to $nN$}
  \FOR{$l = 1$ to $n$}
  \STATE $j=(k-1)nN+(i-1)n+l$
  \STATE
  \(
\Enc(\check{F}^*_j)\gets
\Enc(\check{-1})\otimes\Enc(\check{\Gamma}_i)\otimes\Enc(\check{W}_{il})\otimes\Enc(\check{\Phi}_{k,l})
  \)
  \ENDFOR
  \ENDFOR
  \ENDFOR
  \STATE 
  \(
\Enc(\check{F}^*)\gets\bigoplus_{j=1}^{M}\Enc(\check{F}_j^*)
  \)
  \end{algorithmic}
  \end{algorithm}
\end{figure}

\begin{figure}[t!]
  \begin{algorithm}[H]
  \caption{CKKS-Based Client-Side Confidential FRIT}
  \label{ckks_client_processing}
  \begin{algorithmic}[1]
  \REQUIRE $\Enc(\check{F}^*)$
  \ENSURE $\bar{F}^*_\mathrm{C}$
  \STATE $\check{F}^*\gets\Dec(\Enc(\check{F}^*))$
  \STATE $\bar{F}^*_\mathrm{C}\gets\Dcd_{\gamma_c}(\check{F}^*)$
  \end{algorithmic}
  \end{algorithm}
\end{figure}

\begin{figure}[t!]
  \begin{algorithm}[H]
  \caption{Computation of $\Enc(\check{\Phi}_k)$ for CKKS-based Confidential FRIT}
  \label{ckks_phi_omega_generation}
  \begin{algorithmic}[1]
  \REQUIRE $\Enc(\check{\Psi})$, $\Enc(\lvert\check{\Psi}\rvert^{-1})$, $\Sigma$, $\pk$, and $\evk$
  \ENSURE $\Enc(\check{\Phi}_k)$, $\forall k\in\{1,2,\cdots,(n-1)!\}$.
  \FOR{$k = 1$ to $(n-1)!$}
  \FOR{$i = 1$ to $n$}
  \FOR{$j = 1$ to $n$}
  \STATE Compute $\Enc(\check{\tilde{\Psi}}_{ij})$ from $\Enc(\check{\Psi})$.
  \IF{$\mathrm{sgn}\sigma_k=-1$}
  \STATE
  \(
\Enc(\check{\tilde{\Psi}}_{k,ij})\gets\left(\bigotimes_{l=1}^{n-1}\Enc\left(\check{\tilde{\Psi}}_{l\sigma_k(l)}\right)\right)\otimes\Enc(\check{-1})
  \)
  \ELSE
  \STATE
  \(
    \Enc(\check{\tilde{\Psi}}_{k,ij})
    \gets\bigotimes_{l=1}^{n-1}\Enc\left(\check{\tilde{\Psi}}_{l\sigma_k(l)}\right)
  \)
  \ENDIF
  \IF{$(-1)^{i+j} = -1$}
  \STATE
  \(
    \Enc(\check{\Phi}_{k,ji})
    \gets\Enc(\lvert\check{\Psi}\rvert^{-1})\otimes
    \Enc(\check{-1})\otimes
    \Enc(\check{\tilde{\Psi}}_{k,ij})
  \)
  \ELSE
  \STATE
  \(
\Enc(\check{\Phi}_{k,ji})\gets
\Enc(\lvert\check{\Psi}\rvert^{-1})\otimes\Enc(\check{\tilde{\Psi}}_{k,ij})
  \)
  \ENDIF
  \ENDFOR
  \ENDFOR
  \ENDFOR
  \end{algorithmic}
  \end{algorithm}
\end{figure}

\section{Numerical Examples}\label{sec:exp}
This section demonstrates that the proposed confidential FRIT algorithms, derived from \textbf{Theorem~\ref{thm1}} and \textbf{Proposition~\ref{prop}}, enable gain tuning based on encrypted data through two numerical examples. 

\subsection{Comparison under Unified Security Level}
To ensure a fair comparison between the ElGamal-based and CKKS-based FRIT methods, this study introduces the concept of bit security, which provides a unified measure of security strength applicable across different encryption schemes with appropriately chosen parameters~\cite{barker2020,HomomorphicEncryptionSecurityStandard}.
The parameters of the homomorphic encryption schemes were selected to satisfy the 128-bit security level, meaning that breaking either scheme would require approximately $2^{128}$ computations, which is considered infeasible with current technology.
Specifically, for the ElGamal encryption scheme, the public and private key lengths were determined based on the NIST guidelines\cite{barker2020}, while for the CKKS encryption scheme, the parameter configuration followed the guidelines provided in~\cite{HomomorphicEncryptionSecurityStandard}.

For the ElGamal encryption scheme to achieve 128-bit security, the public and private key lengths were set to 3072 bits and 256 bits, respectively.
In Example 1, the resulting keys were 
$p\approx 3.6\times 10^{924}$, 
$q\approx 1.8\times 10^{924}$, 
$g=2$, and 
$h\approx 2.0\times 10^{924}$, 
with the private key $s$\footnote{$s=\mathtt{0xb0633d52e9a9aa4fcb9563aa0dc8ebe9b7af728c9f2e229332fcbaac4b881b8d}\approx 2.5\times 10^{76}$}.
In Example 2, the resulting keys were 
$p\approx 4.3\times 10^{924}$, 
$q\approx 2.1\times 10^{924}$, 
$g=2$, and 
$h\approx 3.5\times 10^{924}$, 
with the private key $s$\footnote{$s=\mathtt{
0x80af30c2f4d77da1b3d0f28e62e21fff736306e12c6a7d603f24b342de7f0380}\approx 5.8\times 10^{76}$}.  
The sensitivity parameter was set to $\gamma_e=2^{-40}$.
The encryption scheme was implemented using our developed C++ Encrypted Control Library.
For the CKKS encryption scheme, the polynomial degree, the total bit length of the coefficient modulus $\log q$, and the bit distribution of the coefficient modulus $\{q_0,q_1,\dots,q_{20}\}$ were set to 32768, 880, and $\{60,40,\dots,40,60\}$, respectively, where the public, private, and evaluation keys are omitted.
The sensitivity was set to $\gamma_c=2^{-40}$. 
The same encryption parameters were used in both examples. 
The encryption scheme was implemented using the Microsoft SEAL library\footnote{https://github.com/microsoft/SEAL}~\cite{Cheon18}.
All computations were performed on a MacBook Air with an Apple M4 chip and 24 GB of memory, running macOS 15.5. 

\subsection{Example Setups and Tuning Results}
\subsubsection{Example 1}
We consider the same example as one in \cite{Kaneko13}, where the plant and the initial feedback gain are as follows:
\begin{align}
A=\begin{bmatrix} 1 & 1 \\ 0 & -2 \end{bmatrix},\quad
B=\begin{bmatrix} 0 \\ 1 \end{bmatrix},\quad
F_{\mathrm{ini}}=\begin{bmatrix} -0.8 & 2.0\end{bmatrix}.
\end{align}
The initial state $x(0)\in\R^2$ is set to zero.
The sampling period is 10 ms.
The signal $v$ is defined as $v(k)=1$ if $1\leq k\leq 5$ and $v(k)=0$ if $k=0$ or $k>5$.
The datasets $\Gamma$ and $W$ were collected over the steps $k\in\Z_{50}=\{0,1,2,\cdots,49\}$. 
The desired closed-loop characteristics were set to
$H_\mathrm{d}=\begin{bmatrix} H_\mathrm{d1} & H_\mathrm{d2} \end{bmatrix}^\top=
\begin{bmatrix} \frac{1}{z^2-0.5z} & \frac{z-1}{z^2-0.5z}\end{bmatrix}^\top$.
The conventional FRIT method~\cite{soma2004new} returns the state feedback gain, $F^{\ast}=\begin{bmatrix} -0.5000 & 1.5000 \end{bmatrix}$.
Meanwhile, the proposed confidential methods return the following state feedback gains $\bar{F}^*_\mathrm{E}$ and~$\bar{F}^*_{\mathrm{C}}$:
\begin{align}
\bar{F}^*_{\mathrm{E}} &=[-0.5000143158246135311~1.49999487779680240], \nonumber\\
\bar{F}^*_{\mathrm{C}} &= [-0.4999214773402256~1.499762191825162].
\end{align}

\begin{figure}[tb]
\centering
\hspace*{-1ex}
\begin{minipage}[t]{0.243\linewidth}\centering
\includegraphics[width=1.13\linewidth]{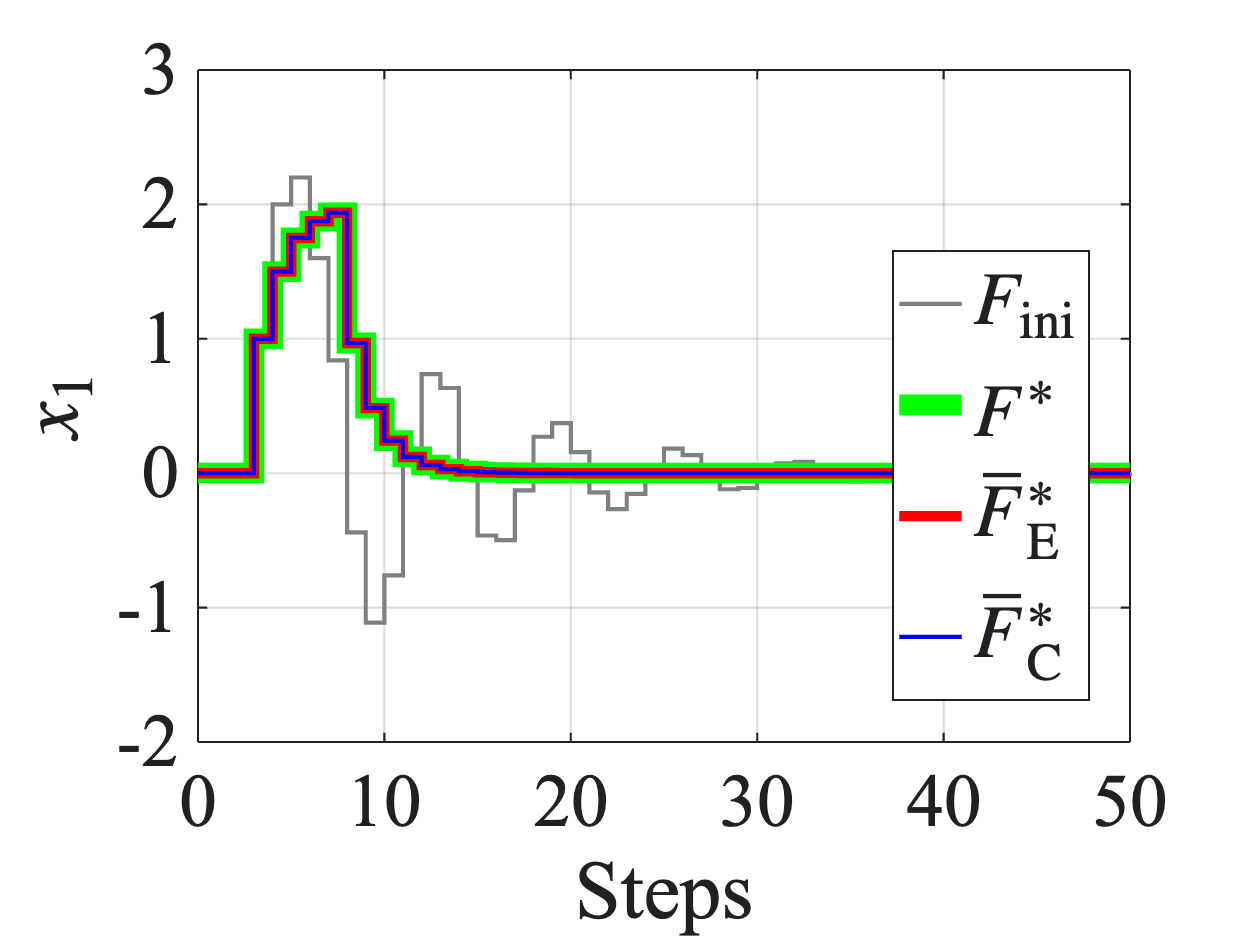}\subcaption{$x_1(k)$}
\end{minipage}
\hspace*{-1ex}
\begin{minipage}[t]{0.243\linewidth}
\includegraphics[width=1.13\linewidth]{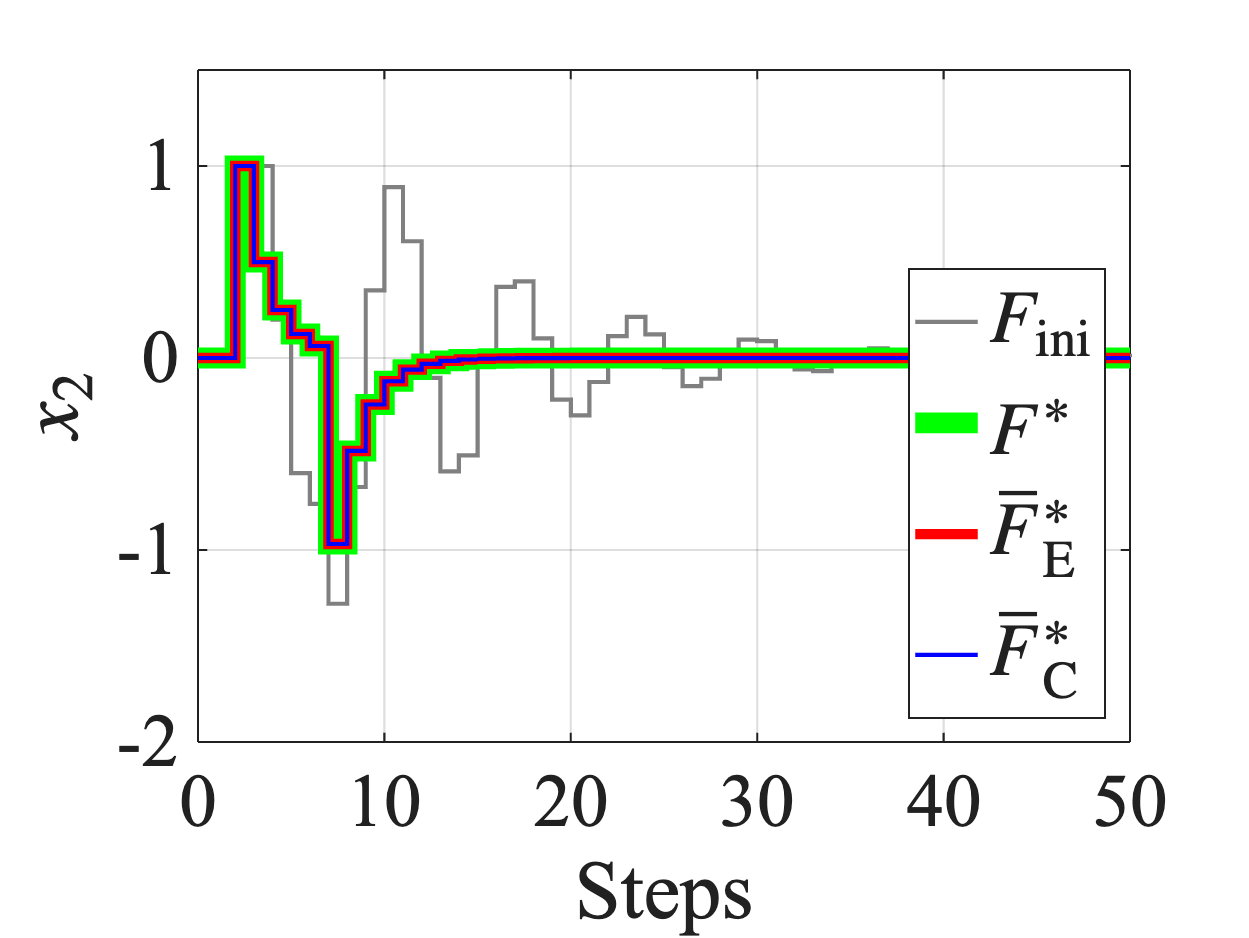}\subcaption{$x_2(k)$}
\hspace*{-1ex}
\end{minipage}
\begin{minipage}[t]{0.243\linewidth}
\includegraphics[width=1.13\linewidth]{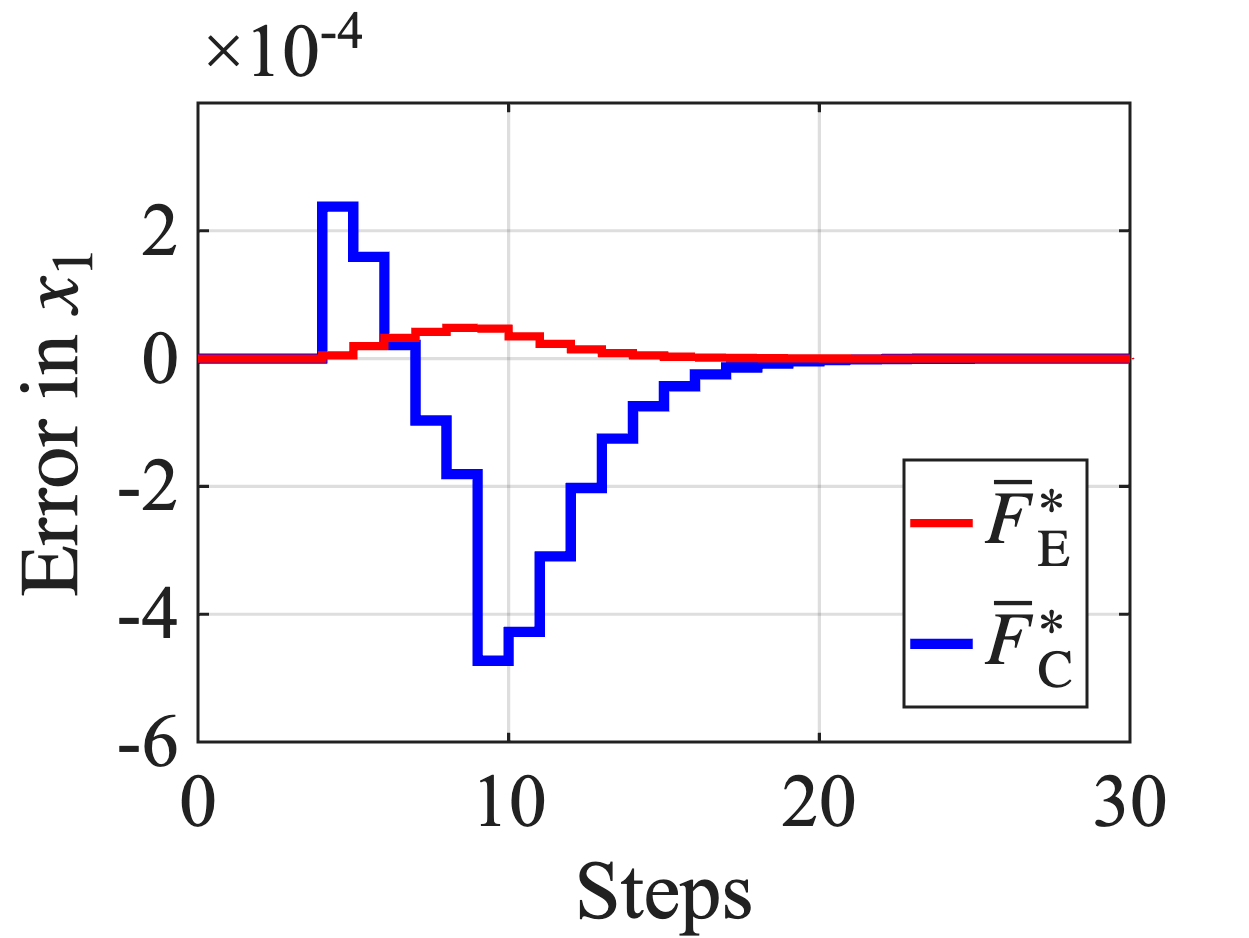}\subcaption{$x_1(k;F^*)-x_1(k;\bar{F}^*_\mathcal{E})$}
\end{minipage}
\hspace*{-1ex}
\begin{minipage}[t]{0.243\linewidth}
\includegraphics[width=1.13\linewidth]{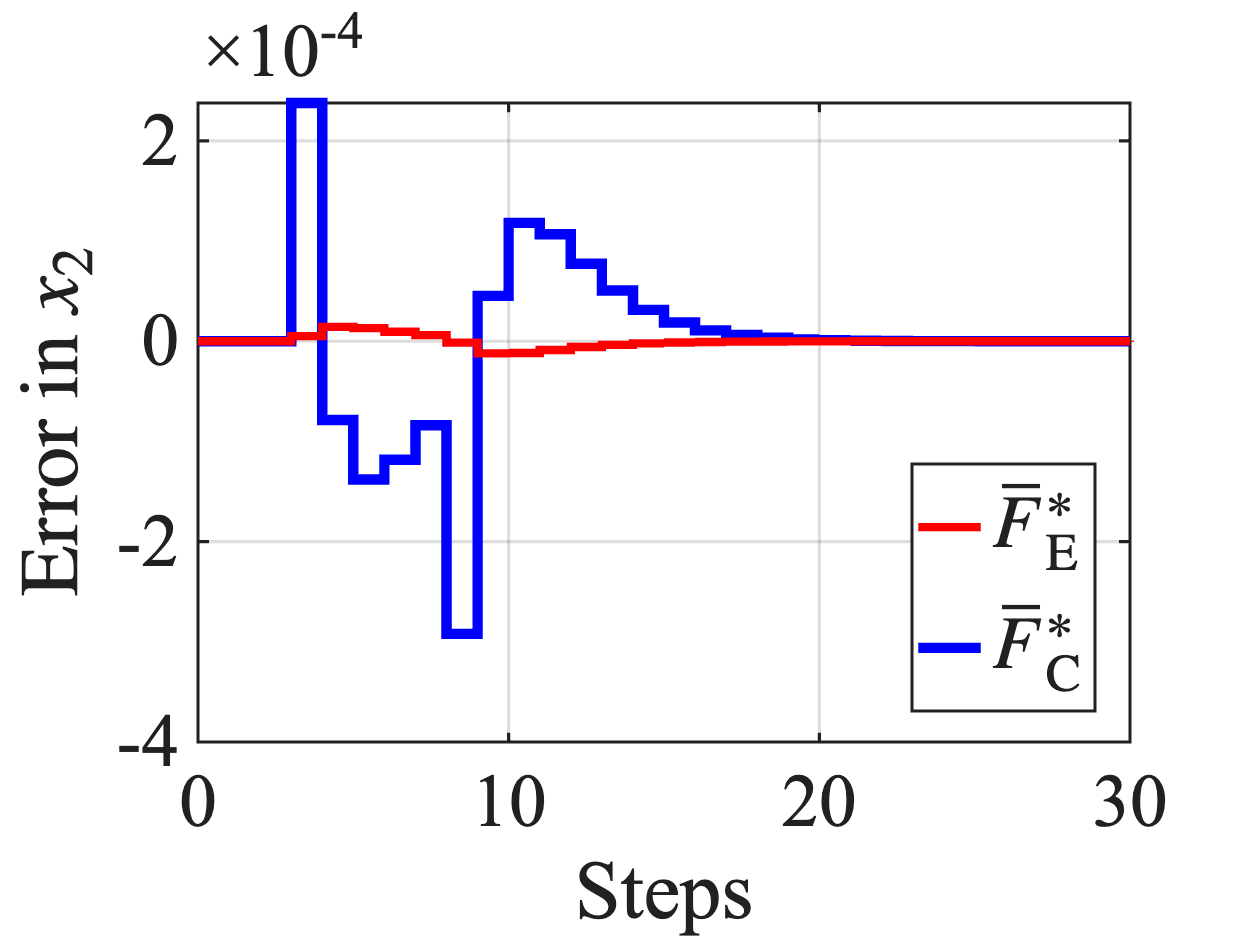}\subcaption{$x_2(k;F^*)-x_2(k;\bar{F}^*_\mathcal{E})$}
\end{minipage}
\caption{Time responses of the state using $F_\mathrm{ini}$, $F^*$, $\bar{F}^*_\mathrm{E}$, and $\bar{F}^*_\mathrm{C}$, and the deviations of the trajectories obtained with $\bar{F}^*_\mathrm{E}$ and $\bar{F}^*_\mathrm{C}$ from that obtained with~$F^*$ in Example~1.}
\label{fig:stateerror}
\end{figure}

\begin{figure}[tb]
\begin{center}
\includegraphics[width=.45\linewidth]{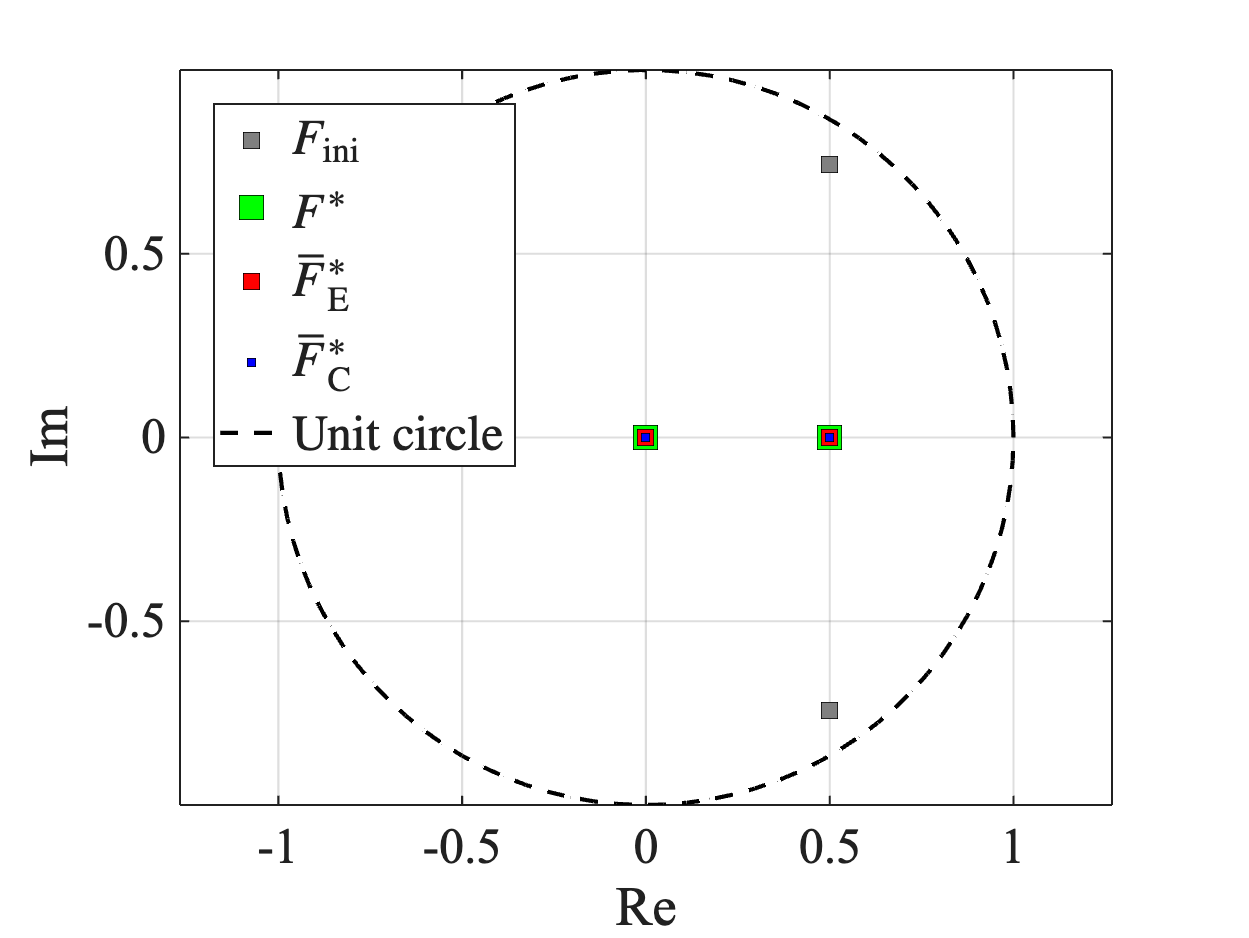}
\caption{Poles of the closed-loop systems using $F_\mathrm{ini}$, $F^*$, $\bar{F}^*_\mathrm{E}$, and $\bar{F}^*_\mathrm{C}$, respectively, in Example~1.
The $\ell_2$-norm of the pole differences between $F^*$ and $\bar{F}^*_\mathrm{E}$, and between $F^*$ and $\bar{F}^*_\mathrm{C}$, are $2.9847\times 10^{-5}$ and $7.4508\times 10^{-4}$, respectively.}
\label{fig:pole}
\end{center}
\end{figure}

The state trajectories are shown in Figs.~\ref{fig:stateerror}(a) and (b), where the feedback gains are set to $F_\mathrm{ini}$, $F^*$, $\bar{F}^*_{\mathrm{E}}$, and $\bar{F}^*_{\mathrm{C}}$, corresponding to the gray, green, red, and blue lines, respectively.
The differences between the state trajectories obtained using the conventional and proposed FRIT methods are shown in Figs.~\ref{fig:stateerror}(c) and (d), illustrating the effects of quantization errors.
The red and blue lines represent the differences using $\bar{F}^*_{\mathrm{E}}$ and $\bar{F}^*_{\mathrm{C}}$, respectively.
The figures show that the ElGamal-based FRIT exhibits smaller deviations from the trajectory obtained using $F^*$ compared to the CKKS-based method.
The magnitude of the differences maintains on the order of $10^{-4}$ over an average of 50 steps, which can be considered negligible.
Furthermore, Fig.~\ref{fig:pole} shows the poles of the closed-loop systems using $F_\mathrm{ini}$, $F^*$, $\bar{F}^*_\mathrm{E}$, and $\bar{F}^*_\mathrm{C}$.
The $\ell_2$-norms of the pole differences between $F^*$ and $\bar{F}^*_\mathrm{E}$, and between $F^*$ and $\bar{F}^*_\mathrm{C}$, are $2.9847\times 10^{-5}$ and $7.4508\times 10^{-4}$, respectively.
These results confirm that the proposed confidential data-driven gain tuning achieves closed-loop poles that are nearly identical to those obtained using the conventional method. 

\subsubsection{Example 2}
We consider the following three-dimensional linear system as a plant:
\begin{align}
A=\begin{bmatrix}     0.9054 & 0.6895 & 0.2246 \\ -0.2246 & 0.2317 & 0.2403 \\  -0.2403 & -0.9455 & -0.2489  \end{bmatrix},\quad
B=\begin{bmatrix} 0.0946 \\ 0.2246 \\ 0.2403 \end{bmatrix}.
\end{align}
The resulting initial state-feedback gain is $F_{\mathrm{ini}}=\begin{bmatrix} 0.12 & -2.37 & -0.82 \end{bmatrix}$.
The initial state $x(0)\in\R^3$ is set to zero.
The sampling period is 1.0 s.
The signal $v$ is defined as $v(k)=1$ if $1\leq k\leq 5$ and $v(k)=0$ if $k=0$ or $k>5$.
The datasets $\Gamma$ and $W$ were collected over the steps $k\in\Z_{30}=\{0,1,2,\cdots,29\}$. 
The desired closed-loop characteristics were set to
$H_\mathrm{d}=\begin{bmatrix} H_\mathrm{d1} & H_\mathrm{d2}  & H_\mathrm{d3}\end{bmatrix}^\top=
\begin{bmatrix} \frac{0.0946z^2 +0.2105z + 0.0342}{z^3 - 0.9803z^2 + 0.4318z - 0.1753} &  \frac{0.2246z^2 -0.1109z - 0.1137}{z^3 - 0.9803z^2 + 0.4318z - 0.1753}
& \frac{0.2403z^2 - 0.5083z + 0.2680}{z^3 - 0.9803z^2 + 0.4318z - 0.1753}\end{bmatrix}^\top$.
Then, the resulting state-feedback gains  $F^{\ast}$, $\bar{F}^*_\mathrm{E}$, and~$\bar{F}^*_{\mathrm{C}}$:
\begin{align*}
F^{\ast} &=\begin{bmatrix}0.1863750 &0.1357217& 0.1833644\end{bmatrix}, \\[.5ex]
\bar{F}^*_{\mathrm{E}} &=\begin{bmatrix}0.1863786 & 0.1357332 & 0.1833673\end{bmatrix}, \\[.5ex]
\bar{F}^*_{\mathrm{C}} &= \begin{bmatrix}0.1863461 & 0.1357010& 0.1833355\end{bmatrix}.
\end{align*}

\begin{figure}[tb]
\centering
\hspace*{-1ex}
\begin{minipage}[t]{0.256\linewidth}\centering
\includegraphics[width=1.15\linewidth]{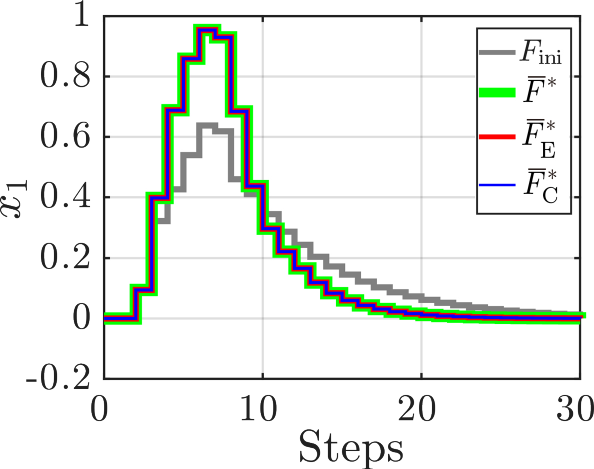}\subcaption{$x_1(k)$}
\end{minipage}
\hspace*{5ex}
\begin{minipage}[t]{0.256\linewidth}
\includegraphics[width=1.15\linewidth]{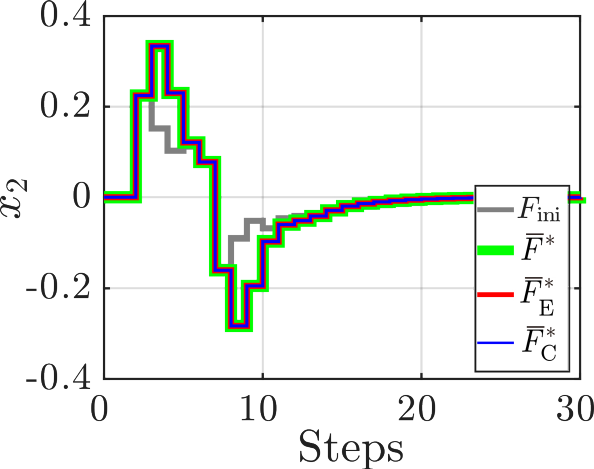}\subcaption{$x_2(k)$}
\end{minipage}
\hspace*{5ex}
\begin{minipage}[t]{0.256\linewidth}
\includegraphics[width=1.15\linewidth]{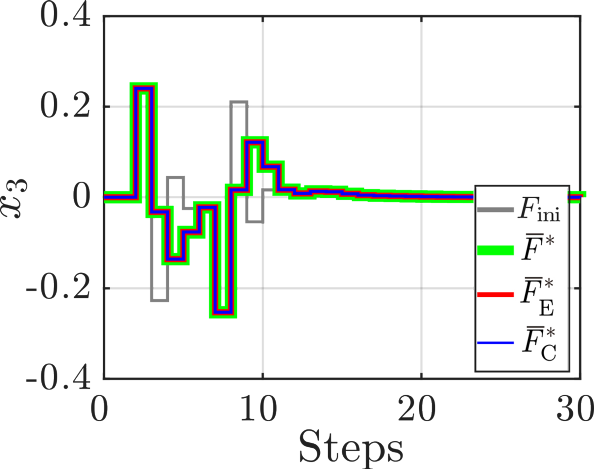}\subcaption{$x_3(k)$}
\end{minipage}\\
\begin{minipage}[t]{0.25\linewidth}
\includegraphics[width=1.15\linewidth]{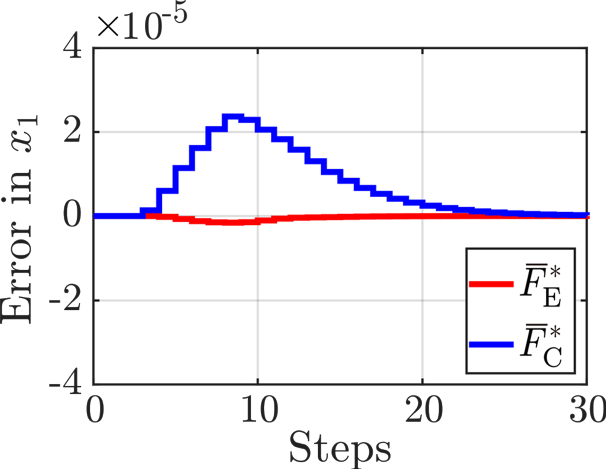}\subcaption{$x_1(k;F^*)-x_1(k;\bar{F}^*_\mathcal{E})$}
\end{minipage}
\hspace*{5ex}
\begin{minipage}[t]{0.25\linewidth}
\includegraphics[width=1.15\linewidth]{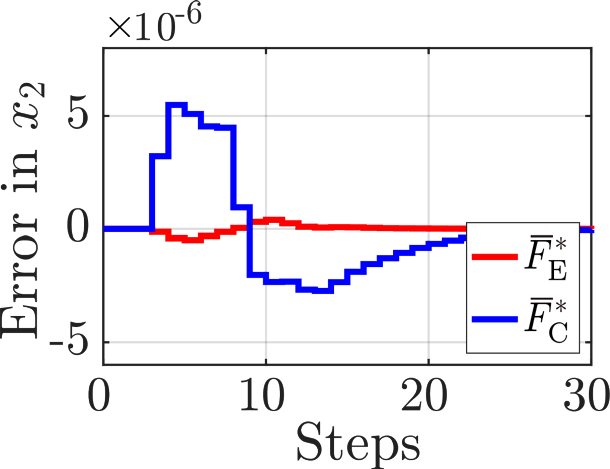}\subcaption{$x_2(k;F^*)-x_2(k;\bar{F}^*_\mathcal{E})$}
\end{minipage}
\hspace*{5ex}
\begin{minipage}[t]{0.25\linewidth}
\includegraphics[width=1.15\linewidth]{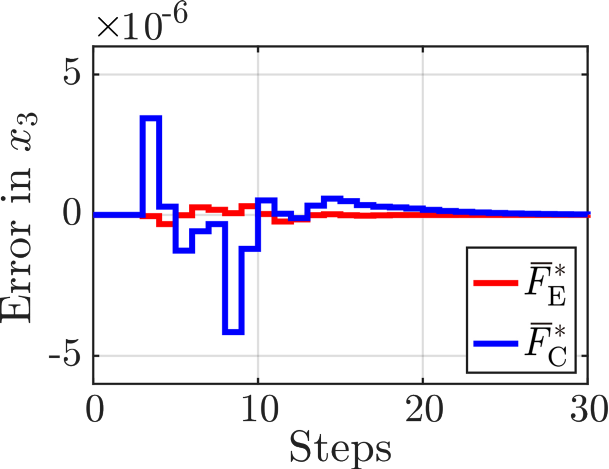}\subcaption{$x_3(k;F^*)-x_3(k;\bar{F}^*_\mathcal{E})$}
\end{minipage}
\caption{Time responses of the state using $F_\mathrm{ini}$, $F^*$, $\bar{F}^*_\mathrm{E}$, and $\bar{F}^*_\mathrm{C}$, and the deviations of the trajectories obtained with $\bar{F}^*_\mathrm{E}$ and $\bar{F}^*_\mathrm{C}$ from that obtained with~$F^*$ in Example~2.}
\label{fig:stateerror2}
\end{figure}

\begin{figure}[tb]
\begin{center}
\includegraphics[width=.45\linewidth]{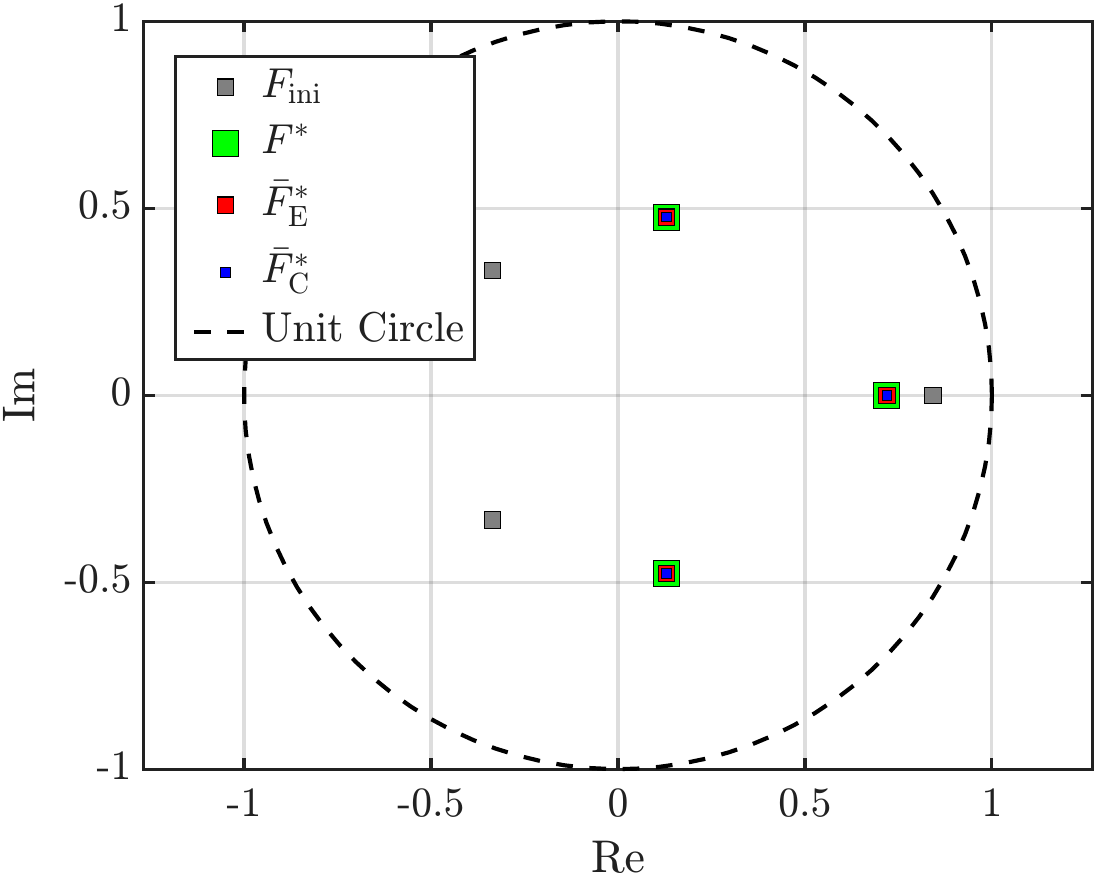}
\caption{
Poles of the closed-loop systems using $F_\mathrm{ini}$, $F^*$, $\bar{F}^*_\mathrm{E}$, and $\bar{F}^*_\mathrm{C}$, respectively, in Example~2. The $\ell_2$-norm of the pole differences between $F^*$ and $\bar{F}^*_\mathrm{E}$, and between $F^*$ and $\bar{F}^*_\mathrm{C}$, are $1.2663\times 10^{-6}$ and  $1.2912 \times 10^{-5}$, respectively.}
\label{fig:pole2}
\end{center}
\end{figure}

The state trajectories are shown in Figs.~\ref{fig:stateerror2}(a)-(c), where the feedback gains are set to $F_\mathrm{ini}$, $F^*$, $\bar{F}^*_{\mathrm{E}}$, and $\bar{F}^*_{\mathrm{C}}$, corresponding to the gray, green, red, and blue lines, respectively.
The differences between the state trajectories obtained using the conventional and proposed FRIT methods are shown in Figs.~\ref{fig:stateerror2}(d)-(f), illustrating the effects of quantization errors.
The red and blue lines represent the differences using $\bar{F}^*_{\mathrm{E}}$ and $\bar{F}^*_{\mathrm{C}}$, respectively.
The figures show that the ElGamal-based FRIT exhibits smaller deviations from the trajectory obtained using $F^*$ compared to the CKKS-based method.
The magnitude of the differences maintains on the order of $10^{-5}$ over an average of 30 steps, which can be considered negligible.
Furthermore, Fig.~\ref{fig:pole2} shows the poles of the closed-loop systems using $F_\mathrm{ini}$, $F^*$, $\bar{F}^*_\mathrm{E}$, and $\bar{F}^*_\mathrm{C}$.
The $\ell_2$-norms of the pole differences between $F^*$ and $\bar{F}^*_\mathrm{E}$, and between $F^*$ and $\bar{F}^*_\mathrm{C}$, are $1.2663\times 10^{-6}$ and $1.2912 \times 10^{-5}$, respectively.
These results confirm that the proposed confidential FRIT achieves closed-loop poles nearly identical to those obtained using the conventional method. 

The computation times at the server side in Fig.~\ref{fig:setup} for the proposed confidential FRITs using the ElGamal and CKKS encryption schemes are summarized in Table~\ref{table:ct}, as evaluated in Example~2.
Each value represents the minimum, median, and maximum time (in milliseconds) measured over ten independent trials.
The table confirms that the ElGamal-based method achieves significantly shorter computation times than the CKKS-based FRIT.
This is primarily due to the lower computational complexity of the homomorphic operations involved in the ElGamal scheme.
These findings indicate that when minimizing computation time is a primary concern, the ElGamal-based FRIT offers a more practical and efficient solution.
Conversely, if resistance to quantum attacks is essential, the CKKS scheme, based on lattice-based cryptography, serves as a more appropriate choice due to its post-quantum security guarantees.

\begin{table}[t!]
  \centering
  \caption{Computation times at the server side for the proposed confidential FRITs in Example~2.}
  \vspace*{-2ex}
  \label{table:ct}
  \begin{tabular}{@{}cccc@{}}
    \midrule
    Scheme & Minimum (ms)& Median (ms) & Maximum (ms) \\ \midrule
    ElGamal-based & $1.032 \times 10^{3}$ & $1.034 \times 10^{3}$ & $1.057 \times 10^{3}$ \\[.5ex]
    CKKS-based    & $2.282 \times 10^{5}$ & $2.338 \times 10^{5}$ & $2.363 \times 10^{5}$ \\ \midrule
  \end{tabular}
\end{table}

\section{Conclusion}\label{sec:con}
This study proposed a confidential FRIT framework and its corresponding algorithms using homomorphic encryption schemes, such as ElGamal or CKKS.
In the numerical examples, with the parameters ensuring 128-bit security, the proposed confidential FRIT algorithms successfully performed gain tuning over encrypted data, resulting in gains approximately identical to those obtained using the conventional FRIT method.
Furthermore, the results confirmed that the closed-loop poles obtained using the ElGamal-based confidential FRIT were closer to the original poles than those obtained using the CKKS-based method. 

In future work, we will extend the proposed algorithms to the tuning of dynamic controllers, integrate confidential FRIT with encrypted control systems to enable secure CaaS, develop detection methods for attacks such as poisoning and deception, and design an efficient matrix inversion technique to reduce data transmission to external servers. 
Additionally, we will explore the broader challenge of securing least squares-based methods, including system identification and optimization-based control.

\appendix
\section{CKKS Encryption Scheme}\label{app:ckks}
This appendix provides a brief introduction to the CKKS algorithms and functions.

\noindent\textbf{Notations}: 
The symbol $\gets$ represents the sampling operator. 
Specifically, $a\gets\mathcal{R}_n$ means that $a$ is a polynomial randomly sampled from $\mathcal{R}_n$. 
For a polynomial $a$, let $\lfloor a\rceil$ be the polynomial obtained by rounding all the coefficients of $a$ to the nearest integer. 
Similarly, for a discrete Gaussian distribution $\chi$ with mean $0$ and variance $\sigma^2$, $b\gets\chi$ means that $b$ represents a polynomial whose coefficients are all sampled from $\chi$. 

\noindent\textbf{Key Generation}: 
Using $\sk$, which is a polynomial over $\mathcal{R}$ with coefficients randomly chosen from $\{-1,0,1\}$, the public and evaluation keys, $\pk=(\pk_0,\pk_1)$ and $\evk=(\evk_0,\evk_1)$, are determined as follows:
$\pk_0=-\pk_1\cdot\sk+e\bmod q_L$ and $\evk_0=-\evk_1\cdot\sk+e'\bmod q_L$,
where $\pk_1\gets\mathcal{R}_{q_L}$, $\evk_1\gets\mathcal{R}_{q_L^2}$, $e\gets\chi$, $e'\gets\chi$, and $q_L=q_0\gamma_c^{-\bar{L}}$,
where $\bar{L}$ is the maximum level.

\noindent\textbf{Encryption}:
Using $\pk$, the encryption function $\Enc(\check{x},\pk)=x^{ct}=(ct_0,ct_1)$ is realized as follows: 
$ct_0=v\cdot\pk_0+e_1+\check{x}\bmod q_L$, $ct_1=v\cdot\pk_1+e_2\bmod q_L$, 
where $v\gets\mathcal{R}_2$, $e_1\gets\chi$ and $e_2\gets\chi$.

\noindent\textbf{Decryption}:
Using $x^{ct}$ and $\sk$, the decryption function $\Dec(x^{ct},\sk)=\check{x}$ is realized, as follows: 
The plaintext $\check{x}$ is calculated based on the divisor $q_l$ updated by the rescaling function:
$\check{x}=ct_0+t_1\cdot\sk\bmod q_l$.
The decryption function produces plaintext with an error term:
$\Dec(x^{ct},\sk)=v\cdot\pk_0+e_1+\check{x}+v\cdot\pk_1+e_2\bmod q_l= \check{x}+v\cdot e+e_1+e_2\cdot\sk\bmod q_l$, where when $\gamma_c$ is sufficiently close to zero, the following approximation holds:
$\check{x}\gamma_c\approx\Dec(\Enc(\check{x},\pk),\sk)\gamma_c$.
This enables high-precision recovery of real-valued data.

\noindent\textbf{Rescaling}:
The rescaling function transforms ciphertext $x^{ct}_l = (ct_{10}, ct_{11}) \in \mathcal{R}^2_{q_l}$ at level $l$ to ciphertext $x^{ct}_{l-1} = (ct_{20}, ct_{21}) \in \mathcal{R}^2_{q_{l-1}}$ at level $l-1$,
where the divisor $q_l$ is updated by multiplying with $\gamma_c$, and for level $l$ $(1\leq l\leq\bar{L})$, $q_l=q_0\gamma_c^l$. 
The rescaling function is defined as: $\Res(x^{ct}_l)=x^{ct}_{l-1}$, where $ct_{20}$ and $ct_{21}$ are calculated as follows:
$ct_{20}=\lfloor ct_{10}\gamma_c\rceil\bmod q_l\gamma_c$ and $ct_{21}=\lfloor ct_{11}\gamma_c\rceil\bmod q_l\gamma_c$.

\noindent\textbf{Homomorphic Operations}:
The addition of ciphertexts $x^{ct}_i$ $(i=1,2)$ is expressed as
$x^{ct}_{\mathrm{add}}=(ct_{30},ct_{31})=x^{ct}_1\oplus x^{ct}_2$, where 
$ct_{30}=ct_{10}+ct_{20}\bmod q_l$ and $ct_{31}=ct_{11}+ct_{21}\bmod q_l$.
As $x^{ct}_1$ and $x^{ct}_2$ are associated with $v_1\gets\mathcal{R}_2$, $v_2\gets\mathcal{R}_2$, and $e_{11}\gets\chi$, $e_{12}\gets\chi$ and $e_{21}\gets\chi$, $e_{22}\gets\chi$, the addition produces:
\begin{align*}
&ct_{30}=(v_1+v_2)\cdot\pk_1+e_{11}+e_{21}+\check{x}_1+\check{x}_2\bmod q_l,\\
&ct_{31}=(v_1+v_2)\cdot\pk_0+e_{12}+e_{22}\bmod q_l.
\end{align*}
This results in a ciphertext containing the sum of the plaintexts.
Furthermore, the multiplication of ciphertexts $x^{ct}_i$ $(i=1,2)$ is expressed as: $x^{ct}_{\textrm{mul}}=x^{ct}_1\otimes x^{ct}_2$,
where
\begin{align*}
x^{ct}_{\mathrm{mul}} &= \Res(x^{ct}_3),\ \ x^{ct}_3=(ct_{30},ct_{31}),\ \ 
ct_{30}=ct_{40}+\lfloor ct_{42}\cdot\evk_0q_L^{-2}\rceil\bmod q_l,\\  
ct_{40}&=ct_{10}\cdot ct_{20}\bmod q_l,\ \ ct_{42}=ct_{11}\cdot ct_{21}\bmod q_l,\\
ct_{31}&=ct_{41}+\lfloor ct_{42}\cdot\evk_1q_L^{-2}\rceil\bmod q_l,\ \ ct_{41}=ct_{10}\cdot ct_{21}+ct_{11}\cdot ct_{20}\bmod q_l.
\end{align*}
The relinearization is performed on $ct_{40}$, $ct_{41}$, and $ct_{42}$ using $\evk$. 
Consequently, the homomorphism regarding addition and multiplication is explained as follows: for arbitrary real numbers $x_i$ $(i=1,2)$, their encoded and encrypted values are $\check{x}_i$ and $x_i^{ct}$, respectively. 
The following calculations hold for $x_i^{ct}$:
\begin{align*}
  \Dec(x_1^{ct}\oplus x_2^{ct},\sk)\approx\check{x}_1+\check{x}_2,\quad  
  \Dec(x_1^{ct}\otimes x_2^{ct},\sk)\approx\check{x}_1\check{x}_2\gamma_c, 
\end{align*}
where $\oplus$ and $\otimes$ represent addition and multiplication operations between encrypted data, respectively.

\section{FRIT for State-Feedback Law}\label{app:frit}
Based on FRIT~\cite{Kan15,Kan16,Kan24}, the details of \eqref{original} are introduced.
Let $x(F)$ and $u(F)$ denote the state and control input associated with the gain $F$, respectively.
The control system designer (client) collects the $N$ signal data points, represented by $\Gamma\in\R^{nN}$ and $W\in\R^{nN\times n}$, where
\begin{align}    
\Gamma&\coloneqq
\begin{bmatrix}\gamma_1 \\ \gamma_2 \\ \vdots \\ \gamma_n \end{bmatrix}\in\mathbb{R}^{nN}, \quad 
  \gamma_j= 
  \begin{bmatrix}
    x_j(k;F_\mathrm{ini})-H_{\mathrm{d}j}u(k;F_\mathrm{ini})\\
    \vdots\\
    x_j(k+N;F_\mathrm{ini})-H_{\mathrm{d}j}u(k+N;F_\mathrm{ini})
  \end{bmatrix}
  \in\mathbb{R}^{N},\\
W&\coloneqq
\begin{bmatrix}w_1 \\ w_2 \\ \vdots \\ w_n\end{bmatrix}\in\mathbb{R}^{nN\times n}, \quad
w_j= 
  \begin{bmatrix}
    H_{\mathrm{d}j}x(k;F_\mathrm{ini})^\top\\
    \vdots\\
    H_{\mathrm{d}j}x(k+N;F_\mathrm{ini})^\top
  \end{bmatrix}
  \in\mathbb{R}^{N\times n},\quad \forall j\in\{1,\cdots,n\}.
\end{align}
The resulting gain $F^*$ obtained from \eqref{original} minimizes the objective function: $J(F)=\|x(F_\mathrm{ini})-H_d \tilde{v}(F)\|_2=||(H(F)-H_d)\tilde{v}(F)||_2$, where the pseudo exogenous signal is given by $\tilde{v}(F)=u(F_\mathrm{ini})-Fx(F_\mathrm{ini})$.
Assuming that the signal $\tilde{v}$ contains certain modes, the minimizer $F^*$ to $J$ is obtained when $H(F)\rightarrow H_d$ is achieved.
Consequently, $J(F^*)=(\Gamma+WF^{*\top})^\top(\Gamma+WF^{*\top})$ holds, and \eqref{original} is derived via the least squares method.
In addition, the above discussion describes the case using closed-loop data; however, FRIT can also be applied to open-loop data. 
For the sake of simplicity, the detailed variable definitions for the open-loop case are omitted.

\section*{Disclosure statement}
No potential competing interest was reported by the author(s).

\section*{Funding}
This work was supported by JSPS KAKENHI Grant Numbers 22H01509 and 23K22779.

\bibliographystyle{tfnlm}
\bibliography{reference}

\end{document}